\documentclass{llncs}
\usepackage{epsfig} 
\usepackage{color}
\usepackage[]{algorithm2e}
\usepackage{stackengine}
\usepackage{dsfont}
  \usepackage{bbm}

   \usepackage{amsmath,amssymb,amsthm}

\begin{document}
\title{Efficient Multilinear Map from Graded Encoding Scheme}
\author{Majid Salimi}
\institute{Faculty of Computer Engineering, University of Isfahan, Isfahan, Iran\\
\email{M.Salimi@eng.ui.ac.ir,MajidSalimi3@gmail.com}}

\maketitle
\begin{abstract}
	Though the multilinear maps have many cryptographic applications, secure and efficient construction of such maps is an open problem. Many multilinear maps like GGH, GGH15, CLT, and CLT15 have been and are being proposed, while none of them is both secure and efficient. The construction of some multilinear maps is based on the Graded Encoding Scheme (GES), where, the necessity of announcing zero-testing parameter and encoding of zero has destroyed the security of the multilinear map.
	Attempt is made to propose a new GES, where, instead of encoding an element, the users can obtain the encoding of an associated but unknown random element. In this new setting, there is no need to publish the encodings of zero and one. This new GES provides the actual functionality of the usual GES and can be applied in constructing a secure and efficient multilinear map and a multi-party non-interactive key exchange (MP-NIKE) scheme. We also improve the MP-NIKE scheme of \cite{Access20} and turn it into an ID-based MP-NIKE scheme.
	
\end{abstract}

\section{Introduction} 
Multilinear map is the extension of bilinear map, where, instead of two elements, it is given many elements as the input. This map is widely applied in many cryptographic applications, like broadcast encryption, attribute-based encryption, indistinguishable obfuscation, and multi-party key exchange. \\
Boneh and Silverberg (2003) introduced the notion of multilinear map and showed that a multilinear map can be applied in constructing efficient broadcast encryption and non-interactive key-exchange, but they did not introduce any concrete construction for multilinear maps therein. They found that the well-known bilinear pairings cannot be generalized in constructing multilinear maps. 
 Gentry et al. (2013) found that a multilinear map can be applied in constructing an Indistinguishable Obfuscation $i\mathcal{O}$ \cite{GGH13b}; moreover, it can be applied in constructing other tools like Pseudo-Random Function (PRF) and Constrained PRF \cite{5Gen}, Attribute-Based Encryption (ABE) \cite{GGHZ} and Functional Encryption (FE).

\subsection{Related Work}
In 2003, Garg et al. proposed the first candidate of the multilinear map (GGH13) based on the Graded Encoding Scheme (GES) in ideal lattices. Coron et al. (2013) proposed another candidate of the multilinear map (CLT13) based on GES and over integers \cite{CLT13}. Both the CLT13 and GGH13 maps are inefficient \cite{Albrecht14}. Langlois et al. (2014) proposed a more efficient version of GGH, named GGHLite \cite{GGHLite}, and then, Albrecht et al. (2014) implemented the GGHLite and made its codes publicly available under an open-source license \cite{Albrecht14}.
Their construction is inefficient in the sense that, computing the value of a 7-linear map takes about 1.75 seconds on a 16-core CPU \cite{Albrecht14}. Hu and Jia (2015) revealed that the GGH13 and GGHLite multilinear maps are not secure \cite{HuJia}. Following this, Cheon et al. (2015) proposed an attack on the CLT13 multilinear map \cite{Cheon15}. In the same year, Gentry, Gorbunov and Halevi proposed a graph-based multilinear map named GGH15 \cite{GGH15}, and immediately, Coron et al. (2014) proposed an attack against GGH15 \cite{CLLT15}. Coron et al. (2015) improved the CLT13 scheme \cite{CLT15}, while, immediately, Minaud and Fouque proposed an attack on their scheme \cite{Minaud15}. All these schemes lack formal security proof. 
Lewi et al. (2016) constructed a framework for CLT multilinear map named 5Gen \cite{5Gen}, which was the most efficient multilinear map \cite{5Gen}, \cite{MZ17} of the date. Ma and Zhandry (2018) proposed a CLT-like scheme and proved that their scheme is secure against known attacks, but they did not prove the security in a formal model \cite{MZ17}. \\
In 2016, Albrecht et al. (2016) proposed a multilinear map from obfuscation  \cite{AlbrechtTCC}, \cite{AlbrechtJoC}, which is not efficient because of employing costly tools like indistinguishable obfuscation and homomorphic encryption.
 \subsection{Our Contribution}
A new GES scheme is presented in this article. In previous graded encoding schemes each element has many encodings,  while in this scheme any element has a unique encoding. In the previous schemes, any code-word could be easily decoded without re-randomization, while, this encoding scheme, even without re-randomization is irreversible. Previous schemes need encodings of zero for re-randomization, while, in this scheme, re-randomization and encodings of zero as its public parameters are not required. Next to these, they are restricted to a pre-defined level $\kappa$ to extract a representation of an element, while the proposed GES scheme can extract the representation of the element in any level of encoding. In this proposed GES, users can directly check the equality of two encoded values, thus, no need for zero-testing parameter. 
In this proposed GES, users can add and multiply the encoded elements. Multiplication of encoded elements will increase the level of encoding similar to that of the available GES schemes. The advantage of this newly proposed GES vs. its counterparts is that users can extract a representation of encoded elements at any arbitrary level of encoding, without being limited to a predefined level $\kappa$. Contrary to the available GES candidates, this proposed GES scheme is practical and efficient in different aspects like the system parameters’ size, encoded elements’ size and computation cost.
After introducing this GES, we will explain how it can be applied in multilinear map and designing MP-NIKE scheme, where the key generation algorithm is not run by a Trusted Third Party (TTP). 

	\subsection{Paper Organization}
The rest of this article is organized as follows: In Section \ref{Preliminaries} the notation, required definitions and the proposed GES general model are introduced. Our proposed GES scheme is described in Section 3. In Section \ref{MMapConstruction}, an efficient multilinear map based on the proposed GES is constructed. An MP-NIKE scheme without key generation algorithm run by TTP is proposed in Section \ref{MP-NIKE}. In Section \ref{ID-MP-NIKE}, the MP-NIKE scheme of \cite{Access20} is converted to an ID-Based scheme. The security of the proposed multilinear map is proved in Section \ref{SecurityProof}, and finally, the paper is concluded in Section 2.
	\section{Preliminaries}
\label{Preliminaries}
The required mathematical background and the GES general model are introduced here.
\subsection {Notations}
The notations applied throughout this article are listed in Table 1.	
\begin{table}
	\label{Notation}
	\caption{  Notations}
	\begin{center}
		\begin{tabular}{|p{40pt}|p{180pt}|}
			\hline
		 \textbf{Symbol} & \textbf{Description}    \\
			\hline	
			$z(xq)$ & A random polynomial of $xq$, such that $z(0)=0$\\
			\hline	
			$r^{(j)}$ & $j$-th bit of the integer $r$\\
			\hline	
			$ p$, $x$ and $q$ & Three large primes\\
			\hline
			$N$ & Modulus of computation\\
			\hline
			$\mathds{G}$ & A multiplicative subgroup of $\mathds{Z}^*_N$\\
			\hline
			$g$ & A generator of $\mathds{G}$\\	
			\hline	
			$r$ &Random integer\\
			\hline
			$W$ & A set of encoded base elements which is used for encoding other elements. \\
			\hline
			$u_i$ & Encoding of $r_i$\\
			\hline
			$b_i$ & Blinded encoding of $r_i$\\
			\hline
			$[x]_j$& The element $x$ is encoded at level $j$ \\
			\hline
			$SP$ & System parameters\\
			\hline
			$msk$ & The KGC's master secret key\\
			\hline
			$e_i$ & Encoding of $y_i$\\
			\hline
			$d_i$ & Blinded encoding of $y_i$\\
			\hline
			$\gamma$ & Security parameter\\
			\hline
			$A$ & The adversary\\
			\hline
			$B$ & The challenger algorithm\\
			\hline
		\end{tabular} \\\setlength\tabcolsep{5pt}		
	\end{center}
\end{table}

\subsection {Graded Encoding Scheme}
The recent constructions of multilinear maps (e.g. the GGHLite multilinear map) use an encoding scheme named Graded Encoding Scheme (GES) \cite{GGHLite} based on ideal lattices. In the available GES, a same element has many different encodings. In this section, the GGHLite graded encoding scheme \cite{GGHLite} is briefly described. The GES encoding scheme involves the following seven algorithms.

\textbf{Instance generation InstGen($1^\lambda$, $1^\kappa$)}: Consider polynomial rings $R=\mathds{Z}[x]/\langle x^n+1\rangle$ and $R_q=R/qR$. The algorithm InstGen is given a security parameter $\lambda$ and the multilinearity parameter $\kappa$ as the input, and generates system parameters $SP$ and one instance of level-1 encoding of $1$ ($[o]_1$) and two instances of level-1 encodings of zero $([z_1]_1, [z_2]_1)$. The $[z_1]_1$ and $[z_2]_1$ will be considered as randomizers.

\textbf{Level-0 sampler samp($SP$)}: In this algorithm, any user chooses an integer $y$ which in GES is defined as a level-0 encoding. In lattice-based GES schemes, the level-0 encoding $y$ is an integer sampled from discrete Gaussian distribution. 

\textbf{Level-1 encoding enc($SP$, $y$)}: Given a level-0 encoding $y$ and $SP$, this algorithm computes $[y]_1=y[o]_1$ and yields the encoded value of $y$ at level $1$, while, by having the publicly known $[o]_1$ computing $[y]_1$ becomes reversible and the value $y$ can be computed. In this case, the user must randomize $[y]_1$ as follows: 
\begin{equation}
[y\acute{]}_1=[{y}]_1+r_1[z_{1}]_1+r_2[z_{2}]_1,
\end{equation}
where $r_1$ and $r_2$ are two random integers (level-0 encodings). Note that the resultant $[y\acute{]}_1$ is still an encoding of $y$ at level 1. Adding different multiples of $[z_1]_1$ and $[z_2]_1$, as their name suggests, is to randomize the encoding, and does not affect the actual value and level of the the encoding. The $\acute{[y]_1}$ is a random encoding of integer $y$ and computing the value of $y$ from $[y\acute{]}_1$ is intractable. In the GES encoding scheme $[y]_1$, $[y\acute{]}_1$ and $[o]_1$ are level-1 encodings.\\

\textbf{Adding encodings add($[a]_l , [b]_l)$}:
For adding the two elements $[a]_l$ and $[b]_l$ together they must be at the same level $l$ and the result is $[a+b]_l$ at level $l$.\\ 

\textbf{Multiplying encodings mult($[a]_{l_1} , [b]_{l_2})$}:
Multiplication of two elements $[a]_{l_1}$ and $[b]_{l_2}$ with any arbitrary level $l_1$ and $l_2$ is possible and the result will be $[ab]_{l_{1}+l_{2}}$ at level $l_1+l_2$.\\

\textbf{isZero($SP$, $p_{zt}$, $[y]_\kappa$)}:
The isZero algorithm is applied in testing whether an encoded element $[y]_\kappa$ (at a pre-defined level $\kappa$) is a valid encoding of zero or not, applicable in equality test of two encoded elements (i.e. if $[y_1]_\kappa-[y_2]_\kappa$ is valid encoding of zero then $y_1=y_2$). The $p_{zt}$ is the zero testing parameter and if $[y]_\kappa.p_{zt}$ is less than specific value ($q^{3/4}$), then $[y]_\kappa$ is a valid encoding of zero, (i.e. $y = 0$).\\

\textbf{ext($SP$, $pzt$, $[y]_\kappa$)}:
For extracting the random representation of an encoded element, it must be at a predefined level $\kappa$, where users can extract it by multiplying it with the \textit{zero testing parameter}. The result is a unique image of $y$ not its actual value. This image of $y$ only depends on its exact value and does not depend on the randomizers and their coefficients. For exact details, refer to \cite{GGHLite}.
\subsection{Definitions }

To describe and prove the security of the proposed scheme, the following definitions should be defined:

\begin{definition} (Subset Sum Problem (SSP)).
\label{DSSP}
Let $W=\{e_1, e_2, \dots,e_m\}$ be a set of $m$ random integers. The SSP is the problem of finding a subset of $W$, where the sum of its elements is equal to a given value $u$. 
\end{definition}

The SSP is an NP-complete problem and the complexity of its best  possible solutions is equal to $O(2^{m/2})$ \cite{KCSSP}, $O(u\sqrt{m})$ \cite{JVSSP} and $O(mc)$ \cite{P99}, where $c$ is the size of the largest element of $W$. 
\begin{definition} \label{MMap}(Multilinear Map). Let $\mathds{G}_1$ and $\mathds{G}_2$ be two cyclic groups of the same order $N$. Consider $g_1, g_2,\dots$ and $g_t$ are generators of $\mathds{G}_1$, and $g_{t+1}$ is a generator of $\mathds{G}_2$. The map $e:\underbrace{\mathds{G}_1\times\mathds{G}_1\times\dots\times\mathds{G}_1}_\text{t times} \rightarrow \mathds{G}_2$ is named a multilinear map if it satisfies the following properties \cite{BS03}.
\begin{enumerate}
	\item For any integer $\beta \in \mathds{Z}$ and generators $g_{i} \in \mathds{G}_1$, $i=  1,\dots,t$, and $g_{t+1} \in \mathds{G}_2$, we have $e(g_{1}^{\alpha}, \dots, g_j,\dots, g_{t}^{\gamma})^{\beta}$ $=$ $e( g_{1}^{\alpha}, \dots, g_{i}^{\beta},\dots, g_{t}^{\gamma})$.
	\item The map $e$ is non-degenerate, i.e. if $g_i \in \mathds{G}_1$, for $i=1,\dots,t $, are all generators of $\mathds{G}_1$, then $g_{t+1}$=$e(g_1,\dots,g_t)$ is a generator of $\mathds{G}_2$.
\end{enumerate}
\end{definition}
\begin{definition} \label{Iso}(Injective function $f$).
	  A function $f:\mathds{S} \rightarrow \mathds{G}$ is injective if it maps each element of $\mathds{S}$ into exactly one distinct element of $\mathds{G}$, where $\mathds{S}$ is a set and $\mathds{G}$ is a group. The size of $\mathds{G}$ must be equal or greater than $\mathds{S}$.
\end{definition}
\begin{definition} \label{Iso}(Isomorph ring).
	The ($\mathds{R}_1, +, \times$) $\cong $ ($\mathds{R}_2, +, \times$) are isomorph rings if there exists a reversible bijective function $f:\mathds{R}_1 \rightarrow \mathds{R}_2$ in a sense that it maps each element of $\mathds{R}_1$ to exactly one element of $\mathds{R}_2$ and each element of $\mathds{R}_2$ is paired to exactly one element of $\mathds{R}_1$. For any elements $a$ and $b$ $\in \mathds{R}_1$ the following is yield.
	\begin{align}
		f(a)+f(b)=f(a+b)\\
		f(a)\times f(b)=f(a\times b)
	\end{align}
\end{definition}

\begin{definition} \label{QRing}(Quotient ring). Let $\mathcal{I}$ be an ideal in $\mathds{Z}$ we define an equivalence relation $\sim$ on $\mathds{Z}$ as below\\
\begin{align}
&a+\mathcal{I}=\{a+r:r \in \mathcal{I}\}\\
&a\sim b  \iff a-b \in \mathcal{I}
\end{align}
We use $\mathds{Z}/\mathcal{I}$ to show the set of all such equivalence classes. The $\mathds{Z}/\mathcal{I}$ is a ring and we have
\begin{align}
(a+I)+(b+I)&=(a+b)+I\\
 (a+I)(b+I)&=(ab)+I
\end{align}
 The multiplicative identity of $\mathds{R}=\mathds{Z}/\mathcal{I}$ is $\bar {1}=(1+I)$ and its zero-element is $\bar {0}=(0+I)=I$.
\end{definition}

\section {The Proposed Graded Encoding Scheme}
Before describing our Graded Encoding Scheme, first, the Yamakawa self-bilinear map should be discussed. The bilinear map is a multilinear map that has two domain inputs ($t=2$ in Definition 4) and a different target group, where its target group and domain groups are different. Yamakawa et al. proposed a new bilinear map, named the self-bilinear map, where the two domain groups and the target group are the same with unknown order \cite {Yamakawa}. If the order of the self-bilinear map group is known (e.g. a safe prime $2p+1$, where $p$ is also prime), then the CDH problem will be easy in this group, because, if $e(g,g)=g_1=g^c \bmod 2p+1$ then $e(g^a,g^b)=g_1^{ab} \bmod 2p+1=g^{abc} \bmod \ 2p+1$. Therefore computing $g^{ab} \bmod 2p+1$ is possible by computing  $e(g^{c^{-2}},e(g^a,g^b))=g^{bc} \bmod 2p+1$, while $g^{c^{-2}}=g^{c^{p-3}} \bmod 2p+1$ \cite{Yamakawa}, \cite{NoteSBM}. The $g^{c^{-2}} \bmod N$ can be computed by at most $O(logp)$ computations since $e(g^{c^i},g^{c^i})=g^{c^{2i+1}}$ and $e(g,g^{c^i})=g^{c^{i+1}}$ \cite{NoteSBM}.
 Yamakawa revealed that in self-bilinear schemes the decisional assumptions like the decisional Diffie-Hellman (DDH) assumption cannot hold and they consider only computational assumptions like CDH problem \cite{Yamakawa}.
A self-multilinear map is defined as follows:
\begin{definition} Self-Multilinear Map: Let $\mathds{R}$ be a ring of the unknown order $xq$ and $\mathds{G}$ be an isomorph cyclic group of the unknown order $xq$. Then the map $e:\underbrace{ \mathds{G}\times\mathds{R}\times\dots \times\mathds{R} }_\text{t \ times}\rightarrow \mathds{G} $ is called a self-multilinear map if it satisfies the following properties:
\begin{enumerate}
\item For any integer $r \in \mathds{S}$, where $\mathds{S}$ is a set of integers, generator $g \in \mathds{R}$ and generator $\acute{g} \in \mathds{G}$ we have
\begin{equation}
 e(\acute{g}, g,\dots, g,\dots, g)^{r} = e(\acute{g}, g,\dots, g^{r},\dots,g)=e(\acute{g}^r, g,\dots, g,\dots,g).
 \end{equation}
\item The map $e$ is non-degenerate; that is, if $g_i \in \mathds{R}$, for $i = 1,\dots,t-1$ are generators of $\mathds{R}$ and $\acute{g}$ is a generator of $\mathds{G}$ then $g_t=e(\acute{g}, g_1,\dots, g_{t-1})$ is also a generator of $\mathds{G}$.
\end{enumerate}
\end{definition}
Let $N=\acute{p}\acute{q}=(xp+1)(2q+1)$ be a modulus, where $q>p$, $\mathcal{I}$ be a principal ideal over $\mathds{Z}$ generated by $xq$ and $z(xq)$ be a random polynomial of $xq$,
then $\mathds{Z}/xq\mathds{Z}=\{0, 1,\dots,xq-1\}$ is considered as a quotient ring. In this article the $(\mathds{R}=\mathds{Z}/\mathcal{I}, +, \times)$ is defined as a ring of order $xq$ with generator $p+z(xq)$. We have $y_1p^2+z(xq) \mod xq=(y_1p^2 \bmod xq)+z(xq)$, so $y_1p^2+z(xq)$ is an element of $\mathds{R}$. 
Let $g$ be a generator of subgroup $\mathds{G_1}$ of $\mathds{Z}^*_N$ of order $pxq$. The $(\mathds{G}, +)$ is defined as a group of order $xq$ with generator $g^p \bmod N$, where the ring $\mathds{R}$ over $+$ and the group $(\mathds{G}, +)$ are isomorph, which makes them behave a similar manner.\\

In the proposed GES, $[y]_1=e=yp+kxq \in \mathds{R}$ is a level-1 public encoding of $y$ and $[y]^{b}_1=d=g^{py} \bmod N \in \mathds{G}$ is the blinded version of this encoding of $y$, where $y \in \{0, 1,\dots, xq-1\}$ is the plaintext and $k$ is a random integer. The exponent of $p$ is encoding level (e.g. $yp^2+z(xq)$ is a level-2 encoding of $y$, where $z(xq)$ is the polynomial of $xq$, thus making $z(0)=0$). The $z(xq)$ is considered as the noise of encoding and does not affect the result of extraction algorithm. Let $[y_{i_1}]^{b}_{1}=d_{i_1}$ be a blinded level-1 encoding of $y_{i_1}$, and $e_{i_2}, e_{i_3}$ be the level-1 encodings of $y_{i_2}$ and $y_{i_3}$, respectively. The blinded level-1 encoding of $y_i$ is referred to as the representation of $y_i$ at level 1. To extract a representation of element $F=y_{i_1}y_{i_2}y_{i_3}$ at level 3, users can simply compute $(d_{i_1})^{e_{i_2}e_{i_3}}=g^{p^3y_{i_1}y_{i_2}y_{i_3}} \bmod N$.
The final result of the extraction algorithm does not depend on the actual polynomial $z(xq)$, which may be added during arithmetic operations on encoded elements; that is $z(xq)$ is just the encoding noise indicating that there exist too many encodings for an integer $y_i$. Because 1 randomizers are not applied, there is no way for generating two different encodings for the same integer. This property prevents many attacks. If one can find three different encodings $u_{i_1}, u_{i_2}$ and $u_{i_3}$ of the same integer $y_i$, he can break the security of the proposed GES scheme by first, applying the Euclidean algorithm on $u_{i_1}-u_{i_2}$ and $u_{i_1}-u_{i_3}$, which is assumed to output a multiple of $xq$, and second, computing the gcd of this output and $u_{i_2}-u_{i_3}$, which would be $xq$. 
Next to this, because in the proposed GES, users can obtain only a unique codeword of an element in a specific level, no zero testing parameter $p_{zt}$ is needed for equality testing.\\

 In this GES, instead of publishing encodings of zero and one as system parameters, the level-1 and blinded level-1 encodings of $m$ random integers are published. The value of $m$ will be determined later. Now assume that users want to obtain an encoding of an element of $\mathds{S}$, where $\mathds{S}=\{2,3,\dots, 2^m\}$ is a set of integers, then, instead of obtaining the encoding of an element $r \in \mathds{S}$ they can compute an encoding of the image of $r \ (\acute{r})$ in the ring $\mathds{R}$ and group $\mathds{G}$. The users have access to description of $\mathds{S}$, with no idea about the real value of the image of $r$ in the $\mathds{R}$ and $\mathds{G}$. 
 Let $W=\{(e_{1},d_{1}), \dots,(e_{m},d_{m})\}\}$, where $e_{i}=[y_i]_1=y_ip+k_ixq \in \mathds{R}$ and $d_{i}=[y_i]^{b}_1=g^{py_i} \bmod N \in \mathds{G}$ for $i \in \{1, \dots, m\}$, are the valid level-1 and blinded level-1 encodings of $y_i$ respectively. The $k_i$'s are the random integers and $\{y_1,\dots, y_m\}$ is a super-increasing set of private integers. For any integer $r_i \in \mathds{S}$, the image of $r_i$ is defined in the ring $\mathds{R}$ and group $\mathds{G}$, symbolized by $\acute{r}_i$, as follow:
\begin{equation}
\acute{r}_i=\displaystyle\sum^{m}_{j=1} y_jr_i^{(j)},
\end{equation}
where $r_{i}^{(j)}$ is the $j$th bit of $r_i$. That is users encode $r$ and they obtain an encoding of the image of $r$ in the ring $\mathds{R}$ and group $\mathds{G}$. The encoding algorithm is explained in the Section 3.1.\\
The $\{y_1,\dots, y_m\}$ must be a super-increasing set, that is, there exists an injective relationship between $r_i$ and $\acute{r}_i$, where $\acute{r}_i=y_1r_{i}^{(1)}+\dots+ y_mr_{i}^{(m)}$. To generate this set in a secure manner, first, the two $80$-bit random integers $y_1$ and $y_2$ must be chosen, where $y_2>y_1$, followed by choosing $y_3>y_2+y_1$, and $y_4>y_3+y_2+y_1$ and so on.\\
If $\{y_1,\dots, y_m\}$ is not a super-increasing set, then there exist no one-to-one relationship between $r_i$ and $\acute{r}_i$. In this case, there is a chance to have the two different integers $r_i$ and $r_j$ with the same images $\acute{r}_i=\acute{r}_j$. However, Because of the large space of $2^m$ and $\mathds{Z}^*_N$, this probability is negligible.\\ 
\begin{definition} 
	\label{GCDH}(GCDH and Ext-GCDH Problems). Given system parameter $SP$ and $t+1$ level-1 encodings $[y_0]_1,[y_1]_1,\dots, [y_t]_1$ as inputs, the Graded-CDH and Extraction Graded CDH problems are defined as follows \cite{GGHLite}.
	\begin{itemize}
		\item \textbf{(Graded CDH Problem)}: 
		Output $[y_1\times \dots \times y_t]_{t-1}$.
			
		\item \textbf{(Ext-Graded CDH Problem)}: 
		Output the extracted string of level-$t$ encoding $[y_1\times \dots \times y_t]_{t}$, which is equal to level-$t$ blinded encoding $[y_1\times \dots \times y_t]^{b}_{t}$.
		\end{itemize}
\end{definition}
  Because of its self-multilinearity property, the Graded-DDH problem is not hold in the proposed GES. Based on Theorem \ref{GCDH}, the Graded-CDH in the proposed GES is as hard as integer factorization. 
\begin{theorem}
	\label{GCDH}
 Computing $[y]_{t-1}=yp^{t-1}+z(xq)$ from $[y]_t=yp^{t}+z(xq)$ (downgrading the level of encodings) is as hard as integer factorization.  	
\end{theorem}
\begin{proof}
Let $N=(2p+1)(2q+1)$ be a RSA modulus, $e_r=(N-1)/2=p+(2p+1)q$

 and $r \in_R Z^*_N$. Suppose that there exists an adversary able to downgrade the level of encoding, then there exists an algorithm $B$, which can factor $N$. The algorithm $B$ sends $re_r=rp+z(q)$ to the adversary. Note that $re_r$ is a valid level-1 encoding. Assume that the adversary is able to compute $r+z(q)$ and send it back to the algorithm $B$, then the algorithm $B$ can compute $z(q)=r+z(q)-r$, which is a polynomial of $q$ and can lead the attacker to factor $N$.
\end{proof}
The Ext-Graded-CDH problem indicates that, it is not possible to compute $[y]_t^b=g^{yp^t} \bmod N$ from $[y]_t=yp^t+z(xq)$. For this computation, the attacker must remove $z(xq)$ from $[y]_t$ or extract the $p$-root of $({g^p})^{yp^t+z(xq)} \bmod N$, which needs the knowledge of factorization of $N$. The representation of an element at different levels of encoding is different. The Ext-Graded-CDH problem is formally proved in Theorem 3.\\
\subsection{Description of the proposed GES}
The proposed GES consists of the following five algorithms:

\begin{enumerate}
\item \textbf{Instance generation} InstGen($1^{\lambda}$). Given security parameter $\lambda$, this algorithm determines system parameters $SP=\{N, \mathds{G}, \mathds{R}, \mathds{S}, g^p, $ $W=\{(e_{1},d_{1}), \dots$ $,(e_{m},d_{m})\}\}$, where $N=(xp+1)(2q+1)$ is a RSA modulus, $\mathds{G}$ is a subgroup of $\mathds{Z}^*_N$ of order $xq$, $g^p$ is a generator of $\mathds{G}$, $\mathds{R}$ is a ring of order $xq$ as defined in previous section and $e_{i}=[y_i]_1=y_ip+k_ixq \in \mathds{R}$, $d_{i}=[y_i]^{b}_1=g^{py_i} \bmod N \in \mathds{G}$ for $i \in \{1, \dots, m\}$ are the valid level-1 and blinded level-1 encodings of $y_i$'s. The $k_i$'s are the random integers and $\{y_1,\dots, y_m\}$ is a super-increasing set of private integers. The Key Generation Center (KGC) keeps $\{y_1,\dots, y_m\}$ private.\\
\item \textbf{Obtaining a level-1 encoding of an element} Level-1-enc($SP$). Any user chooses uniformly random integer $r_i$ and then generates a level-1 encoding ($[r_i]_1=e_i$) and a blinded level-1 encoding ($[r_i]^{b}_{1}=b_i$) of an integer $r_i \in \mathds{S}$ by computing the following equation:
\begin{align}
u_{i}=& r_{i}^{(1)}e_{1}+\dots  +r^{(m)}e_{m}\\  \nonumber
     =& ((r_{i}^{(1)}) y_{1}+\dots+(r_{i}^{(m)} ) y_m)p+((r_{i}^{(1)} ) k_1+\dots+(r_{i}^{(m)} ) k_m)xq\\ \nonumber 
     =& p\acute{r}_i+z(xq),
\end{align}
\begin{align}
b_{i}=& d_{1}^{r_{i}^{(1)}}\times \dots \times d_{m}^{r_{i}^{(m)}}\\\nonumber =&g^{p(y_1(r_{i}^{(1)})+ \dots+ y_m(r_{i}^{(m)}))} \bmod N\\ \nonumber
=&g^{\acute{r}_ip} \bmod N,
\end{align}
where $b_i \in \mathds{R}$ and $u_i \in \mathds{G}$, are the encoding of $\acute{r}$ in $\mathds{R}$ and $\mathds{G}$ respectively.
 The user $i$ publishes $u_i$ as level-1 encoding of the unknown value $\acute{r}_i=\displaystyle\sum^{m}_{j=1}y_jr_i^{(j)}$ and keeps $b_i$ secret for himself. It is obvious that obtaining $r_i$ from $u_i=\displaystyle\sum^{m}_{j=1}r_{i}^{(j)}e_j$ requires solving an instance of the SSP problem which is assumed to be intractable.\\

\item \textbf{Adding encodings} Add($\alpha, \beta$). The addition procedure is given the two blinded encodings ($\alpha, \beta$)=($b_i, b_j$) or the two non-blinded encodings ($\alpha, \beta$)=($u_i, u_j$) with arbitrary level $\kappa$ as input, and will produce another element with the same level through the following equation:
\begin{equation}
u_i+u_j=[r_i]_{\kappa}+[r_j]_{\kappa}=(\acute{r}_i+\acute{r}_j)p^{\kappa}+z(xq)=[r_i+r_j]_{\kappa}
\end{equation}
\begin{equation}
b_ib_j=[r_i]^{b}_{\kappa}[r_j]^{b}_{\kappa}=g^{(\acute{r}_i)p^{\kappa}}g^{(\acute{r}_j)p^{\kappa}}=g^{(\acute{r}_i+\acute{r}_j)p^{\kappa}} \bmod N=[r_i+r_j]^{b}_{\kappa},
\end{equation}
where $z(xq)$ is the encoding noise and does not affect the final result. Users are not allowed to perform any operation on elements of $\mathds{S}$.\\
 
\item \textbf{Multiplying encodings} Mul($[r]_{\kappa_i}, [r_j]_{\kappa_j}$). The multiplication procedure is given two non-blinded encodings $u_i=[r_i]_{\kappa_i}$ and $u_j=[r_j]_{\kappa_j}$ with arbitrary levels $\kappa_i$ and $\kappa_j$ as the input and multiplies them as follows:
\begin{equation}
u_iu_j=[r_i]_{\kappa_i} [r_j]_{\kappa_j}=(\acute{r}_i \acute{r}_j)p^{\kappa_i+\kappa_j}+z(xq)=[r_ir_j]_{\kappa_i+\kappa_j},
\end{equation}
where the result is a level-$(\kappa_i+\kappa_j)$ encoding of $r_{1}r_{2}$. Users are not allowed to perform any operation on elements of $\mathds{S}$.\\ 

\item \textbf{Extraction} Ext($N, [r_{i}]^{b}_{\kappa_{i}}, [r_{j}]_{\kappa_{j}}$). This algorithm is given an encoding $u_j=[r_{j}]_{\kappa_{j}}$ with any arbitrary level $\kappa_{j}$ and a blinded level-$\kappa_{i}$ encoding $b_i=[r_{i}]^{b}_{\kappa_{i}}$ as the input and then extracts the representation (blinded encoding) of $[r_{i}.r_{j}]^{b}_{\kappa_{i}+\kappa_{j}}$ at level $\kappa_{i}+\kappa_{j}$.
\begin{align}
(b_{i})^{u_{j}}&=g^{p^{\kappa_{i}+\kappa_{j}}\acute{r}_{i}\acute{r}_{j}} \bmod N=[r_ir_j]^{b}_{\kappa_{i}+\kappa_{j}}
\end{align}  

\end{enumerate}
Let $a_1,a_2 \in \mathds{S}$ be two integers and $a_3=a_1+ a_2$, then $[a_3]_{1} \ne [a_1]_{1} + [a_2]_1$ (e.g. $[2]_1+[4]_1\ne [6]_1$). Because, the elements of $\mathds{S}$  is mapped to $\mathds{R}$ and $\mathds{G}$ by injective functions $f_1:\mathds{S} \rightarrow \mathds{R}$ and $f_2:\mathds{S} \rightarrow \mathds{G}$ (i.e. functions $f_1$ and $f_2$ are not  morphisms), then $f(u\times v) \ne  f(u)\times f(v)$ and $f(u+v) \ne f(u)+f(v)$. As a result, no operation on $\mathds{S}$ is defined, in the first place.
This property does not cause any obstacle for the applications of the resultant GES such as key exchange and multilinear map, because, in these applications, extracting the value of an encoded element is not intended, and users just wants to securely obtain a unique key by using a public tool.\\

\textbf{Parameter setting.} For $\lambda=80$ bit security the modulus of computation must be at least 1024-bit. Computing $r_i$ from $u_i$ is an instance of the SSP problem. The most efficient algorithms for solving the SSP problem have the complexity $O(mc)$, where $c$ is a size of the largest member of $W$, $O(u\sqrt{m})$ and $O(2^{m/2})$. At $m=2\lambda=160$, it is necessary to assure that $u \ge c > 2^{80}$. Here, $N$ is a RSA modulus, $p$, $x$ and $q$ are chosen to be at least 256-bit primes and all of the members of $W$ are at least 512-bit integers, thus $u \ge c > 2^{511}$.  
\section{Multilinear Map}
\label{MMapConstruction}
The proposed GES can be used to construct an efficient multilinear map. In this system, a trusted third party runs InstGen($1^\lambda$) to generate system parameters $SP$, and then it will announce it publicly. Then any user can generate his encoding, do various operations on encodings, and run a multilinear map in any arbitrary multilinearity-degree.\\
\textbf{Instance Generation($1^\lambda$).}
The instance generation algorithm takes the security parameter $\lambda$ as input and then runs InstGen($1^\lambda$). Finally, it publishes $SP=\{N, \mathds{G}, \mathds{R}, \mathds{S}, g^p, W=\{(e_{1},d_{1}), \dots,(e_{m},d_{m})\}\}$ as public parameters.

\textbf{Element Encoding($SP)$.} Given system parameters $SP$ as input, user $i$ generates a level-1 encoding $(u_i=[r_i]_1)$ and a blinded level-1 encoding $(b_i=[r_i]^{b}_{1})$ of an arbitrary element $r_i$ by running Level-$1$-enc($SP$) algorithm. The user keeps $b_i$ secret and publicly publishes $e_i$.

\textbf{Group Operation($SP,\alpha,\beta$).}
This algorithm takes the system parameters $SP$ and two encodings $(\alpha,\beta)$ as input, where $(\alpha,\beta)=([r_i]_\kappa,[r_j]_\kappa)$ or two encodings $(\alpha,\beta)=([r_i]_{\kappa_i},[r_j]_{\kappa_j})$ or two blinded encodings $(\alpha,\beta)=([r_i]^{b}_\kappa,[r_j]^{b}_\kappa)$. Users can compute $[r_i]^{b}_{\kappa}[r_j]^{b}_{\kappa}=[r_i+r_j]^{b}_{\kappa}=g^{p(\acute{r}_i+\acute{r}_j)} \bmod N$ and  $[r_i]_{\kappa}+[r_j]_{\kappa}=[r_i+r_j]_{\kappa}=(\acute{r}_i+\acute{r}_j)p^\kappa+z(xq)$ as well as multiple of two non-blinded level-$\kappa_i$ and level-$\kappa_j$ encodings by computing $[r_i]_{\kappa_i}[r_j]_{\kappa_j}=[r_ir_j]_{\kappa_i+\kappa_j}=(\acute{r}_i\acute{r}_j)p^{\kappa_i+\kappa_j}+z(xq)$, where the result is at level $\kappa_i+\kappa_j$ and $z(xq)$ is the noise of encoding. 

\textbf{Multilinear Map($SP, [r_{i_1}]^{b}_{\kappa_{i_1}}, [r_{i_2}]_{\kappa_{i_2}}, \dots, [r_{i_{t+1}}]_{\kappa_{i_{t+1}}}$).} 
This algorithm is given the system parameters $SP$, $t$ encodings $\{u_{i_2}=[r_{i_2}]_{\kappa_{i_2}}, \dots, u_{i_{t+1}}=[r_{i_{t+1}}]_{\kappa_{i_{t+1}}}\}$ with level $\{\kappa_{i_2}, \dots, \kappa_{i_{t+1}}\}$ and one blinded level-$\kappa_{i_1}$ encoding $b_{i_1}=[r_{i_1}]^{b}_{\kappa_{i_1}}$ as input and generates a blinded level-$(\kappa_{i_1}+ \dots+ \kappa_{i_{t+1}})$ encoding (representation) of $[r_{i_1}\dots r_{i_{\kappa+1}}]_{\kappa_{i_1}+ \dots+ \kappa_{i_{t+1}}}$ as below
\begin{align}
e(b_{i_{1}}, u_{i_2},\dots ,u_{i_{t+1}})=&(b_{i_{1}})^{\prod\limits^{t+1}_{j=2}u_{i_j}} \bmod N\\ \nonumber
=& g^{p^{(\kappa_{i_1}+ \dots+ \kappa_{i_{t+1}})}\prod\limits^{t+1}_{j=1}u_{i_j} } \ mod \ N\\ \nonumber
=&[r_{i_1}\dots r_{i_{t+1}}]^{b}_{\kappa_{i_1}+ \dots+ \kappa_{i_{t+1}}}
\end{align}
The following lemma proves that the proposed multilinear map is non-degenerate.
\begin{lemma}
  	The proposed self-multilinear map is non-degenerate.
\end{lemma}
\begin{proof}
Let the order of $\mathds{G}$ and $g^p$ be a semi-prime integer $xq$, then for any integer $y \in \mathds{Z}-\{xq,2xq,\dots\}$, $(g^p)^{y}$ is also a generator of $\mathds{G}$. Next to this, because, there is not any two integers $y_1$ and $y_2$ such that $y_1y_2=0 \bmod xq$, the proposed multilinear map is non-degenerate.
\end{proof}
The proposed multilinear map needs an off-line trusted party to perform a trusted setup. However, the proposed self-multilinear map does not need KeyGen phase as in MP-NIKE schemes, and it is not limited to a predefined number of users. Because of self-multilinearity property, the output of the self multilinear map can be used as an input for another self-multilinear map.

\section{Efficient Multi-party Non-interactive Key Exchange (MP-NIKE) scheme without key generation algorithm run by TTP}
\label{MP-NIKE}
The proposed GES-based MP-NIKE scheme does not require an on-line key generation server since any user can generate his public/private keys. Suppose $k+1$ users aim to compute a shared key. The proposed MP-NIKE scheme has the following three algorithms.
\begin{enumerate}
	\item \textbf{Setup($1^\lambda$)}. This algorithm takes as input the security parameter $\lambda$ and runs  \textbf{InstGen($1^{\lambda}$)}.
	
	\item \textbf{Publish($SP$)}. In this algorithm, given the system parameters $SP$, every user runs \textbf{Level-1-enc($SP$)}. So, at the end of this algorithm, any user, for example the $i$-th user, will receive $k$ encodings $u_1=[r_{1}]_1,\dots,u_{i-1}=[r_{i-1}]_1,u_{i+1}=[r_{i+1}]_1,\dots, u_{k+1}=[r_{{k+1}}]_1$ from other $k$ users. These encodings are regarded as other users' public keys. Moreover, each user, for example the $i$-th user, generates a blinded level-1 encoding $b_{i}=[r_{i}]^{b}_{1}$ which is regarded as his private key.	
	\item \textbf{ShareKey($N, \{[r_1]_1,\dots, [r_{i-1}]_1, [r_{i+1}]_{1},\dots, [r_{{k+1}}]_1 \}, [r_{i}]^{b}_1$).}
This algorithm is given the system parameters, $k$ level-1 encodings $\{[r_1]_1,\dots, [r_{i-1}]_1, [r_{i+1}]_{1},\dots, [r_{{k+1}}]_1 \}$ and a blinded level-1 encoding $b_{i}=[r_{i}]^{b}_{1}$ and outputs the representation of $r_{1}\times r_{2}\times \dots\times r_{{k+1}}$ at level $k$. In the proposed MP-NIKE scheme, users are not limited to a predefined level $K$ and the representation of encoded element can be extracted in any arbitrary level.
\begin{align}
e(b_i, u_{1},\dots, u_{i-1}, u_{i+1},\dots ,u_{{k+1}})=&(d_{i})^{\prod\limits^{k+1}_{j=1,j\ne i}u_{j}} \bmod N\\ \nonumber
=& g^{p^{k+1}\prod\limits^{k+1}_{j=1}\acute{r}_{i_j} } \ mod \ N
\end{align}
For example, suppose that at the end of \textbf{Publish($SP$)} the user $1$ receives three level-1 encodings $u_{2}=y_{2}p+k_{2}xq \bmod N$, $u_{3}=y_{3}p+k_{3}xq$ and $u_{4}=y_{4}p+k_{4}xq$. It wants to obtain a representation of $y_{1}y_{2}y_{3}y_{4}$ at level 4. So, it computes $d_{1}^{u_{j}}$, where $u_{j}=u_{2}u_{3}u_{4}$. Note that the result does not depend on the coefficient of $xq$.\\
\end{enumerate}

\section{ID-Based MP-NIKE}
\label{ID-MP-NIKE}
In ID-based cryptography, the public key of users is obtained directly from the identity of users. Fortunately, in the MP-NIKE scheme of \cite{Access20} any $logN$-bit random integer $r$ can be a valid public key and the KGC has a trapdoor to compute its respective private key, while others cannot do it. The TTP can use algorithm ID-KeyGen($SP, ID_i$) to compute the private key of user $i$.\\

\textbf{ID-KeyGen($N, g, msk, ID_i$).} This algorithm takes the modulus of computation, generator $g$, master secret key $msk=\{p,x, q\}$ and identity of user $i$ as input and after user authentication, it generates a valid private key for user $i$. Let $N=\acute{p}\acute{q}=(xp+1)(2q+1)$ be a $n$-bit  modulus, $H:\{0, 1\}^* \longrightarrow \{0, 1\}^n$ be a secure cryptographic hash function and $h_i=H(ID_i)$ be the public key of user $i$. Any integer bigger than Frobenius number ($pqx-p-qx$) can be represented by linear combination of $p$ and $xq$ \cite{Sylveste}. So, any $n$-bit random integer $r$ can be valid public key of MP-NIKE scheme of \cite{Access20}.

The KGC computes the private key of user $i$ as below
\begin{align}
h_i p^{-1} \bmod xq=&  y_{h_i}pp^{-1}+k_{h_i}xqp^{-1} \bmod xq\\ \nonumber
=&y_{h_i}p p^{-1} \bmod xq\\ \nonumber
=&y_{h_i} \bmod xq.
\end{align}  

Now the algorithm ID-KeyGen generates the private key for user $i$ as $g^{py_{h_i}} \bmod N$.



\section{Security Analysis}
\label{SecurityProof}
The security of this proposed GES scheme is proved in this Section. First, in Theorem 2, the intractability of decompositioning an encoding $u=[r]_1$ to $e_{1} r^{(1)}+\dots +r^{(j)}e_{j}+\dots +r^{(m)}e_{m} $ and recovering $r$ is proved.
\begin{theorem}
	The adversary that has access to system parameters $SP=\{N, \mathds{G},$ $ \mathds{R}, \mathds{S}, g^p, W=\{(e_{1},d_{1}), \dots,(e_{m},d_{m})\}\}$ and a non-blinded level-1 encoding $u=[r]_1=e_{1} r^{(1)}+\dots +r^{(j)}e_{j}+\dots +r^{(m)}e_{m} $ cannot compute $r$.
\end{theorem}

\begin{proof}
	Suppose that there is an adversary $A$ which can compute $r$ from $u=[r]_1$, then there is an algorithm $B$ which can solve an instance of the SSP problem. Algorithm $B$ is given a non-super-increasing set of $m$-bit integers $\{e_1,\dots, e_m\}$ as input and must find all the subsets whose sum of elements is equal to $u$. The algorithm $B$ generates a RSA modulus $N=(xp+1)(2q+1)$, choose a generator $g \in \mathds{Z}^*_N$ and then computes 
	\begin{align}
	y_{i}&=e_{i} p^{-1} \bmod xq\\
	d_i&=g^{py_{i}} \bmod N.
	\end{align}
	Note that any integer larger than Frobenius number $(pxq-p-xq)$ can be a valid public encoding in the proposed GES. Finally the algorithm $B$ sends $SP=\{N, \mathds{G}, \mathds{R}, \mathds{S}, g^p, W=\{(e_{1},d_{1}), \dots,(e_{m},d_{m})\}\}$ and $u$ to the adversary $A$. 
	If the adversary $A$ can computes $r$ such that $u=r^{(1)}e_{1}+\dots+r^{(m)}e_{m}$ and sends it back to $B$. Then algorithm $B$ outputs $r^{(1)},\dots,r^{(m)}$ $(r)$ as the solution of the SSP problem. 
\end{proof}
Theorem 3 proves that the security of the proposed scheme is equal to the MP-NIKE scheme of \cite{Access20}.
\begin{theorem}
The security of the proposed GES is equal to the MP-NIKE scheme of \cite{Access20}.
\end{theorem}
\begin{proof}
	Suppose that there exists an adversary $A$ that can breaks the security of the proposed GES scheme, then there exists an algorithm $B$ that can break the security of the MP-NIKE scheme of \cite{Access20}. Suppose that there is an algorithm $B$ which is given the public parameter $PP= \{ N, H(\cdot), g^p\}$ and $m$ public/private key pairs $W=\{ (e_1, d_1), \dots, (e_m, d_m)\}$ as input and it must compute the shared key of group $W^*=\{u_{1},\dots, u_{s}\}$, where $u_i$ is a valid public key of \cite{Access20} MP-NIKE scheme. The algorithm $B$ gives $SP=\{N, \mathds{G}, \mathds{R}, \mathds{S}, g^p, W=\{(e_{1},d_{1}), \dots,(e_{m},d_{m})\}\}$ as the system parameters and $W^*=\{u_{1},\dots, u_{s}\}$ as a valid level-1 encodings to the attacker $A$. If the adversary $A$ can computes Ext-Graded CDH($u_{1}, u_{2}, \dots, u_{s}$) and return {\tiny }it to the algorithm $B$, then the algorithm $B$ can solve the inputted problem. 
\end{proof}

\section{Conclusion}
In this article, a new efficient GES is introduced, where users are simply able to generate an encoding of the image of an element, without re-randomization.  In this scheme, the users are not able to generate several encodings of the same element, while the noise of encoding is always fixed, thus, the zero testing parameter and re-randomization are not necessary. The extraction algorithm can run on any encoding with any arbitrary level. This GES can be applied in constructing a multilinear map. This GES multilinear map is secure, efficient and practical in different applications, like the multi-party non-interactive key exchange without the key generation algorithm run by TTP. The MP-NIKE scheme of \cite{Access20} is improved here and turned to an ID-based scheme.
\label{Conclution}

\end{document}